\documentclass[12pt]{article}
\usepackage{natbib}
\usepackage[ left=1in, top=1in, right=1in, bottom=1in]{geometry}
\usepackage{xcolor,graphicx,bm,bbm,colonequals,amsmath,amssymb,url,amsthm}
\usepackage{array,tabularx,multirow}
\usepackage{enumitem}
\usepackage[font={footnotesize}]{caption,subcaption}
\usepackage{color}
\usepackage{mathtools}

\bibpunct[, ]{(}{)}{;}{a}{,}{,}

\usepackage{verbatim,mathrsfs,colonequals}
\usepackage{cancel}
\usepackage[hidelinks]{hyperref}
\usepackage{amsfonts,setspace}

\newtheoremstyle{propstyle} 
    {3mm}                    
    {1mm}                    
    {\itshape}                   
    {}                           
    {\scshape}                   
    {.}                          
    {.5em}                       
    {}  

\theoremstyle{propstyle}

\usepackage{ifthen}
\newboolean{comments} 
\setboolean{comments}{true} 
\newcommand{\comm}[1]{\ifthenelse{\boolean{comments}}{{\color{red}[[ #1 ]]}}{}}

\newtheorem{theorem}{Theorem}

\newtheorem*{definition*}{Definition}

\newcommand{\real}{\mathbb{R}}

\newcommand{\diag}{diag}

\newcommand{\GP}{GP}

\newcommand{\matern}{\mathcal{M}}

\newcommand{\bc}{\mathbf{c}}

\newcommand{\bs}{\mathbf{s}}

\newcommand{\by}{\mathbf{y}}

\newcommand{\bD}{\mathbf{D}}

\newcommand{\bfgamma}{\bm{\gamma}}

\newcommand{\bfSigma}{\bm{\Sigma}}

\newcommand{\sloc}{\bs}
\newcommand{\eloc}{\tilde{\bs}}
\newcommand{\sphere}{\mathbb{S}}
\newcommand{\domain}{\mathcal{D}}
\newcommand{\lon}{l}
\newcommand{\lat}{L}

\title{Locally anisotropic covariance functions on the sphere}

\author{Jian Cao\thanks{Department of Statistics, Texas A\&M University} \and Jingjie Zhang\footnotemark[1] \and Zhuoer Sun\footnotemark[1] \and Matthias Katzfuss\footnotemark[1] \thanks{Corresponding author: \texttt{katzfuss@gmail.com}}}

\date{}

\begin{document}

\maketitle

\begin{abstract}
Rapid developments in satellite remote-sensing technology have enabled the collection of geospatial data on a global scale, hence increasing the need for covariance functions that can capture spatial dependence on spherical domains. We propose a general method of constructing nonstationary, locally anisotropic covariance functions on the sphere based on covariance functions in $\real^3$. We also provide theorems that specify the conditions under which the resulting correlation function is isotropic or axially symmetric. For large datasets on the sphere commonly seen in modern applications, the Vecchia approximation is used to achieve higher scalability on statistical inference. The importance of flexible covariance structures is demonstrated numerically using simulated data and a precipitation dataset.
\end{abstract}


\vspace{4mm}

\section{Introduction}

Traditionally, geostatistical analysis relied on approximating small or regional spatial domains as flat subsets of $\real^2$. However, since the deployment of satellites in the collection of global data, there is an increasing demand for covariance functions that are valid on spheres. In this paper, we aim to propose a new family of spherical covariance functions, defined over the unit $2$-sphere $\sphere = \{\eloc \in \real^3: \|\eloc\| =1\}$, which are able to capture non-stationary features commonly observed in geostatistical datasets. 

For processes defined over $\sphere$, two different distance measures are commonly used, namely the great-arc (or great-circle) distance, which measures the distance ``going along the surface of the sphere'', and the Euclidean or chordal distance that ``pierces through the sphere.'' The relationship between great-arc distance and chordal distance on $\sphere$ is given by:
\begin{equation}
\label{equ:greatchordal}
 r = 2 \sin (\theta/2),
\end{equation}
where $r$ is the chordal distance between two points on $\sphere$, and $\theta$ is the corresponding great-arc distance. Although a function of chordal distance is naturally a function of great-circle distance, finding a valid correlation function directly based on great-arc distance is not trivial \citep[e.g.,][]{Jones1963} due to the curvature of $\sphere$. Most well-known covariance functions are valid (i.e., positive definite) on $\real^d$, for $d \geq 1$, only under Euclidean or Mahalanobis distance, yet may become invalid under great-arc distance \citep[e.g.,][]{Gneiting2012}. For example, the Mat\'ern covariance function is only valid under great-arc distance if its smoothness is no greater than 0.5 \citep{Gneiting2012}.

\citet{huang2011validity} summarized the validity of commonly used covariance functions under great-arc distance, most of which focused on isotropic covariance functions. \citet{Gneiting2012} further developed characterizations and constructions of isotropic positive definite functions on spheres, and proved that subject to a natural support condition, many isotropic positive-definite functions on the Euclidean space $\real^3$ allow for the direct substitution of the chordal distance by the great-arc distance on the sphere. \citet{ma2012stationary,ma2015isotropic} constructed a family of isotropic covariance functions with polynomials. Similarly, \citet{du2013isotropic} designed isotropic variogram functions on spheres using an infinite sum of the products of positive definite matrices and ultraspherical polynomials. \citet{alegria2021f} proposed the $\mathcal{F}$ family of isotropic covariance functions that parameterizes the differentiability of the Gaussian field, which was shown to outperform the Mat\'ern covariance function using chordal distance.

Popular approaches to modeling nonstationary or anisotropic processes on the sphere are reviewed in \citet{jeong2017spherical}, including the differential operators \citep{jun2007approach,Jun2008, jun2014matern},
spherical harmonic representations \citep{stein2007spatial},
stochastic partial differential equations \citep{lindgren2011explicit},
kernel convolutions \citep{heaton2014constructing}, and multi-step spectrum methods \citep{castruccio2013global,castruccio2014beyond, castruccio2016compressing}. \citet{Hitczenko2012} investigated the theoretical properties of a class of anisotropic processes on the sphere from applying first-order differential operators to an isotropic process. \citet{blakeparametric} proposed adaptions of three existing families of stationary spherical covariance functions, namely the construction from Stieltjes functions \citep{menegatto2020positive}, the $\mathcal{F}$ family \citep{alegria2021f}, and the spectral adaptive approach \citep{emery2021twenty}, to obtain nonstationarity. The above approaches typically require intricate design to be positive-definite and differentiable and are less intuitive for practitioners compared with the straight-forward covariance construction in $\real^3$ using the Euclidean distance. 

Here, we follow the idea of \citet{Yaglom1987}, restricting a valid covariance function in $\mathbb{R}^3$ to $\sphere$ under the chordal distance, which is guaranteed to be valid. \citet{Guinness2016} showed that results from using chordal and great-arc distances are often indistinguishable. This makes intuitive sense, in that when two points on the sphere will tend to be (almost) uncorrelated when they are far apart and will have approximately equal chordal distance and great-arc distance (i.e., $\sin(\theta/2) \approx \theta/2$) when they are close. 

Motivated by this observation, we propose a family of nonstationary, locally anisotropic covariance functions on the sphere based on the locally anisotropic covariance functions for Euclidean space proposed in \citet{Paciorek2006}. Similar ideas were also discussed in \citet{Katzfuss2011d} and \citet{Knapp2012}. We will introduce the properties of our general covariance parameterization and specific parameterizations that lead to isotropic or axially symmetric covariance structures to suit various geostatistical applications. For large datasets on the sphere (e.g., with more than $10^4$ points), straightforward computation of the Gaussian log-likelihood is too expensive for statistical inference, for which we use the Vecchia approximation \citep{Vecchia1988} of Gaussian processes (GPs) in our numerical studies. 

The remainder of this article is organized as follows. Section \ref{sec:review_nonsta} reviews a nonstationary correlation function on $\real^d$. In Section \ref{sec:cov_sphere}, we construct classes of nonstationary covariance functions on the sphere, and provide theorems that specify the conditions for isotropic and axially symmetric covariance structures. Section \ref{sec:Vecchia} reviews the Vecchia approximation for large spatial datasets. In Sections~\ref{sec:simulation} and \ref{sec:real_app}, we use simulated data and a precipitation dataset from a physical model to highlight the advantage of our flexible nonstationary covariance structure. Section \ref{sec:conclusions} concludes. Proofs are provided in the Appendix. Code can be found at \url{https://github.com/katzfuss-group/sphere-local-aniso-cov/}.

\section{Locally anisotropic covariance functions \label{sec:review_nonsta}}

In this section, we briefly review an intuitive construction for nonstationarity proposed in
\citet{Paciorek2006} based on any isotropic correlation function, denoted by $\rho$, in $\real^d$ for all $d \in \mathbb{N}$. Specifically, the nonstationary correlation function is composed as:
\begin{align}
    \label{equ:general_cov_design}
    \rho_{NS}(\bs_i,\bs_j) &= c(\bs_i,\bs_j) \rho\big( q(\bs_i,\bs_j) \big), \\
    \label{equ:mahalanobis}
    q(\bs_i,\bs_j) &= \{ 2 (\bs_i - \bs_j)' (\bfSigma(\bs_i) + \bfSigma(\bs_j) )^{-1} (\bs_i - \bs_j) \}^{1/2}, \\
    \label{equ:cov_const_term}
    c(\bs_i,\bs_j) &= |\bfSigma(\bs_i)|^{1/4} |\bfSigma(\bs_j)|^{1/4} | (\bfSigma(\bs_i) + \bfSigma(\bs_j))/2|^{-1/2},
\end{align}
where the positive-definite $d \times d$ matrix $\bfSigma(\bs_i)$ is the local anisotropy matrix describing the spatially varying rotation and scaling, $q(\bs_i,\bs_j)$ is the Mahalanobis distance with respect to the average anisotropy matrix $(\bfSigma(\bs_i) + \bfSigma(\bs_j) )/2$, and $c(\bs_i,\bs_j)$ is the normalization term. The anisotropy matrix $\bfSigma(\bs)$ at each location needs to be positive definite matrices in $\real^{d \times d}$, to which we assign the local rotational and scaling effects at the location $\bs$. Hence, spatially-varying anisotropy is achieved in the covariance structure represented by Equations~\eqref{equ:general_cov_design}, \eqref{equ:mahalanobis}, and \eqref{equ:cov_const_term}.

This nonstationarity design combines different local anisotropic correlation structures into a valid global correlation function, leading to greater model expressiveness. Under the assumption that the anisotropy matrix $\bfSigma(\bs)$ varies smoothly across the domain, the differentiability of $\rho_{NS}$ follows that of the underlying stationary covariance function $\rho$ \citep{Paciorek2006}. When $\rho$ is chosen as the isotropic Mat\'ern covariance function, one may also vary the smoothness parameter across the domain to achieve different local smoothness; see Section~\ref{subsec:nonstationary_matern} for a detailed description.


\section{Classes of nonstationary covariance functions on the  sphere\label{sec:cov_sphere}}

Given the intuitive construction for nonstationarity in Section~\ref{sec:review_nonsta}, an important question is how to parameterize the anisotropy matrix $\bfSigma(\bs)$ to better capture the local covariance structure, which can be largely problem dependent. In this section, we consider one general parameterization of $\bfSigma(\bs)$ that represents local rotation and scaling on a tangent plane of the sphere and probe the conditions for achieving the special cases of isotropic and axially symmetric covariance functions.

\subsection{Construction of the covariance functions\label{subsec:construct}}

Without loss of generality, assume that $\sphere$ is centered at the origin in $\real^3$, $(0,0,0)$, and that the intersection of the prime meridian and the equator, denoted by $\bc \colonequals (0,0)$ ($0^\circ$ longitude, $0^\circ$ latitude), is located on the x-axis (i.e., it has the Euclidean coordinates $\tilde{\bc} \colonequals (1,0,0)$). Figure~\ref{fig:sphere} shows the part of the sphere that lies in the first (positive) octant of a Cartesian coordinate system, including the origin and $\bc$. The Euclidean coordinates $\eloc=(x,y,z)$ of any point $\sloc = (\lon,\lat)$ with longitude $\lon$ and latitude $\lat$ on $\sphere$ are given by:
\begin{equation*}
 x  = \cos(\lat) \cos(\lon), \quad y  = \cos(\lat) \sin(\lon), \quad z  = \sin(\lat).
\end{equation*}

\begin{figure}
\centering\includegraphics[width=.6\textwidth]{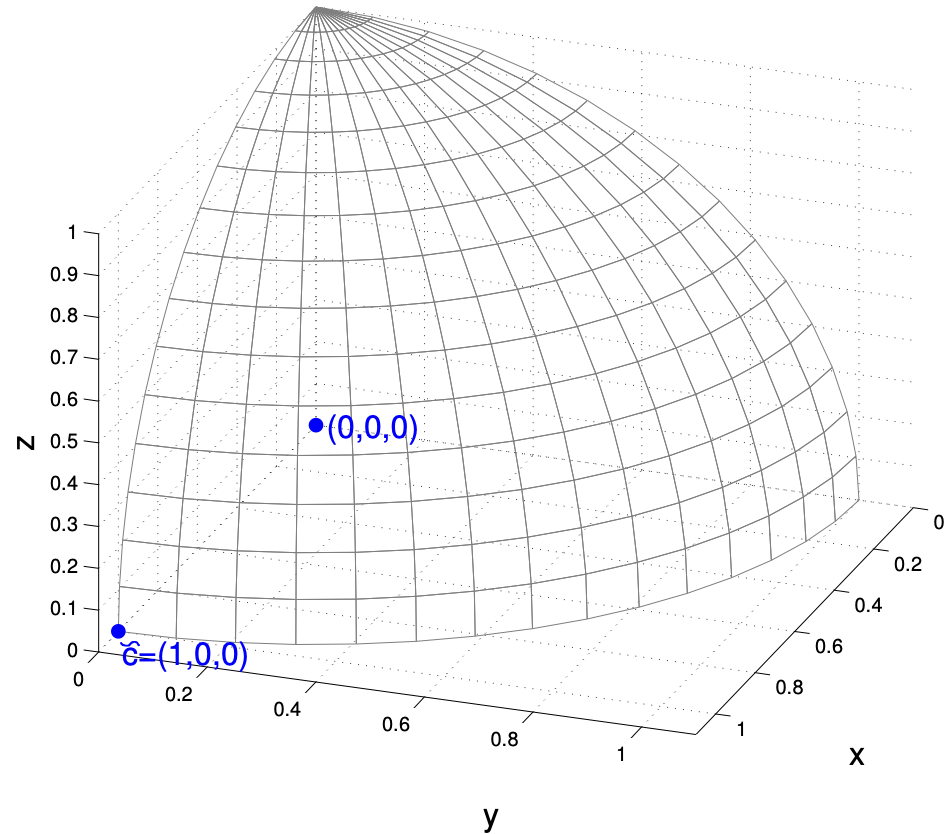}
\caption{\label{fig:sphere}The part of a unit sphere centered at the origin that lies in the first octant of the Cartesian coordinate system, where all coordinates are positive. The origin and the point $\bc$ referred to in the text are shown in blue.} 
\end{figure}

For $d$-dimensional Euclidean space, we can parameterize $\bfSigma(\bs)$ using $d$ scaling parameters and $d-1$ rotation parameters \citep[see, e.g.,][]{Banerjee2008}. Although $\sphere$ ``lives'' in $\mathbb{R}^3$, the surface of $\sphere$ is (locally) a two-dimensional space (i.e., the tangent plane) at any point $\bs \in \sphere$, indicating a parameterization with only two local scaling and one local rotation parameters. Consider the tangent plane at $\bc$, which is the $(y,z)$-plane, as shown in Figure \ref{fig:sphere}. We use $\gamma_1(\bc) >0$ and $\gamma_2(\bc) > 0$ as the scaling parameters in the $y$ and $z$-directions, respectively, and $\kappa(\bc) \in [0, \pi/2)$ as the rotation parameter, whose collective effect is shown in Figure \ref{fig:rotationillus}. Specifically, the scaling matrix at $\bc$ is a diagonal matrix $\bD(\bfgamma) \colonequals \diag \{1, \gamma_1,\gamma_2 \}$ and the rotation matrix that rotates the $(y,z)$-plane at $\tilde{\bc}$ about the $x$-axis is given by:
\[
   \mathcal{R}_x(\kappa) \colonequals \left( \begin{array}{ccc} 1 &0 &0\\ 0 & \cos \kappa &-\sin \kappa\\ 0 & \sin \kappa  & \cos \kappa \end{array} \right).
\]
\begin{figure}
\centering\includegraphics[width=.5\textwidth]{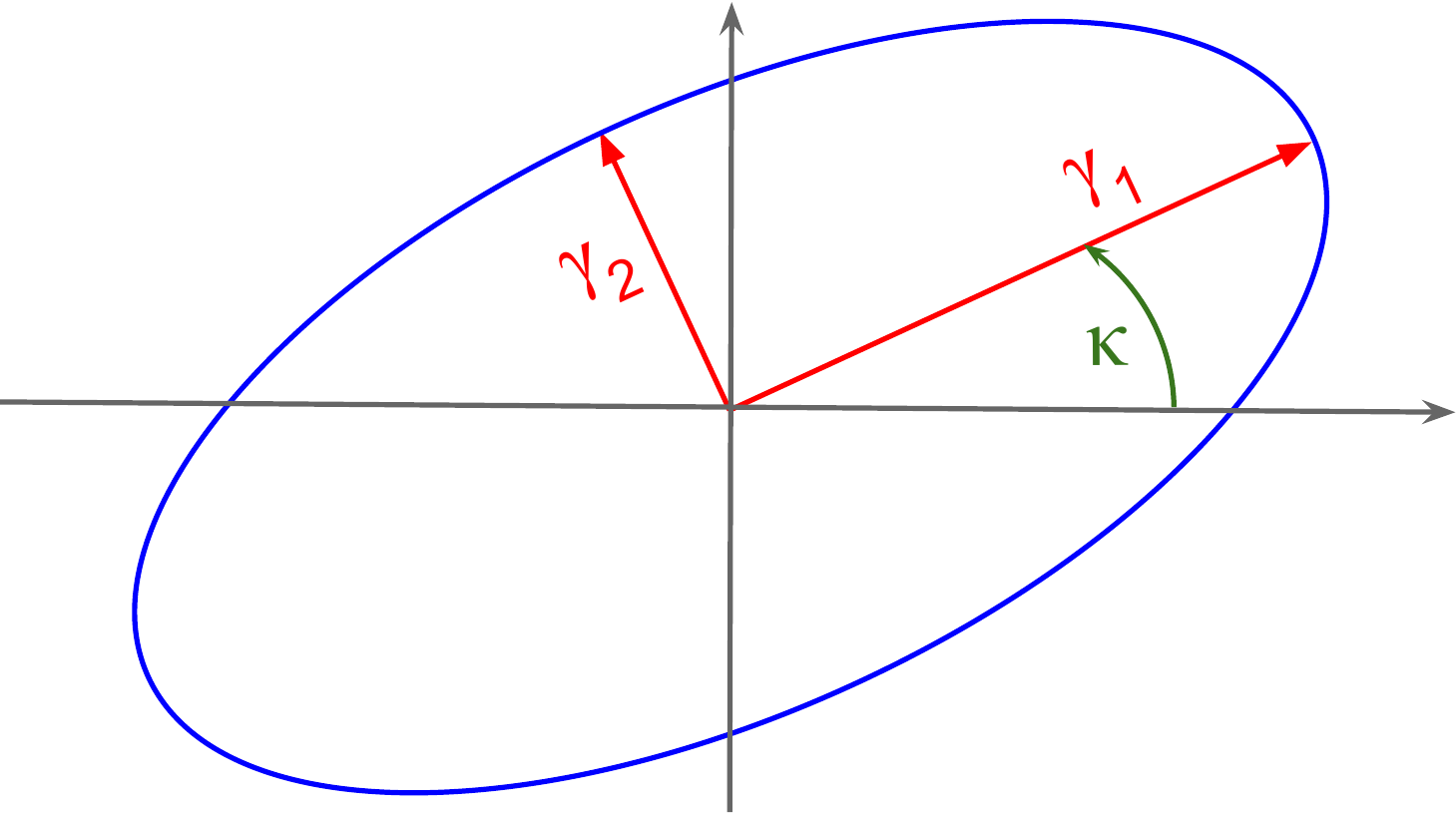}
\caption{\label{fig:rotationillus}Illustration of the scaling parameters $\gamma_1$ and $\gamma_2$ and the rotation parameter $\kappa$ at reference point $\tilde\bc$, along with an oval representing the resulting correlation contour} 
\end{figure}
The anisotropy matrix at the reference point $\tilde{\bc}$ (spherical coordinates $\bc$) is then given by:
\begin{equation}
\label{third:refanisotropy}
 \tilde{\bfSigma}(\bc) \colonequals \mathcal{R}_x(\kappa(\bc)) \bD(\bfgamma(\bc)) \mathcal{R}_x(\kappa(\bc))'.
\end{equation}

Consider an arbitrary location $\sloc$ on the $\sphere$, whose Euclidean coordinates are denoted by $\eloc$. The way we construct the anisotropy matrix for $\sloc$, denoted by $\bfSigma(\sloc)$, is to left- and right-multiply \eqref{third:refanisotropy} by the rotation matrix $\mathcal{R}(\sloc)$ such that $\eloc' \bfSigma(\sloc)^{-1} \eloc = \eloc' (\mathcal{R}(\sloc) \tilde{\bfSigma}(\sloc) \mathcal{R}(\sloc)')^{-1} \eloc =  \tilde{\bc}' \tilde{\bfSigma}(\sloc)^{-1} \tilde{\bc}$. Here, $\tilde{\bfSigma}(\sloc)$ is parameterized by $\kappa(\sloc)$ and $\bfgamma(\sloc)$, similar to \eqref{third:refanisotropy}, denoting the local rotation and scaling at $\sloc$ measured at the reference point $\bc$. To rotate $\sloc$ to $\bc$, we can rotate $\eloc$ about the $y$-axis and $z$-axis by $-\lat$ and $-\lon$, respectively, which is equivalent to left-multiplication of $\tilde{\bc}$ by $\mathcal{R}_y(-\lat) \mathcal{R}_z(-\lon)$, where
\[
\mathcal{R}_y(\theta) \colonequals \left( \begin{array}{ccc} \cos \theta &0 & \sin \theta\\ 0 & 1 & 0\\ -\sin \theta & 0 & \cos \theta \end{array} \right) \quad\quad \textrm{and} \quad\quad
\mathcal{R}_z(\theta) \colonequals \left( \begin{array}{ccc} \cos \theta &-\sin\theta & 0\\ \sin \theta & \cos\theta & 0\\ 0 & 0 & 1 \end{array} \right).
\]
Hence, $\mathcal{R}(\sloc)$ and $\bfSigma(\sloc)$ are defined as $\mathcal{R}_z(\lon) \mathcal{R}_y(\lat)$ and $\mathcal{R}_z(\lon) \mathcal{R}_y(\lat) \tilde{\bfSigma}(\sloc) \mathcal{R}_y(\lat)' \mathcal{R}_z(\lon)'$, respectively. Therefore, at an arbitrary location $\sloc$, we have:
\begin{align*}
    \eloc' \bfSigma(\sloc)^{-1} \eloc = \tilde{\bc}' \tilde{\bfSigma}(\sloc)^{-1} \tilde{\bc} = \tilde{\bc}' \Big(\mathcal{R}_x(\kappa(\sloc)) \bD(\bfgamma(\sloc)) \mathcal{R}_x(\kappa(\sloc))'\Big)^{-1} \tilde{\bc}.
\end{align*}

The anisotropy matrix $\bfSigma(\sloc)$ achieves nonstationarity through introducing local rotation and scaling. One can further increasing the nonstationarity through assuming heterogeneous variances in the domain $\sigma^2(\sloc) >0$, with which the covariance function between two locations $\sloc_i$ and $\sloc_j$ amounts to:
\[
C(\sloc_i,\sloc_j) = \sigma(\sloc_i)\sigma(\sloc_j)\rho_{NS}(\sloc_i,\sloc_j),
\]
where $\rho_{NS}(\sloc_i,\sloc_j)$ is defined by \eqref{equ:general_cov_design} to \eqref{equ:cov_const_term} through the anisotropy matrix $\bfSigma(\sloc)$.

\subsection{Properties}

The approach above provides a parameterization of $\bfSigma(\sloc)$ in terms of two spatially varying ranges, $\bfgamma(\sloc)$, and one spatially varying rotation, $\kappa(\sloc)$, which can in turn be parameterized in suitable ways for different applications. In this section, we provide conditions on $\bfgamma(\sloc)$ and $\kappa(\sloc)$ such that the resulting correlation function in \eqref{equ:general_cov_design} is isotropic or axially symmetric; proofs of the theorems are included in Appendix \ref{app:proofs}.

An isotropic covariance function is a function of only distance between two locations. Due to the one-to-one relationship between the great-arc and the chordal distances in \eqref{equ:greatchordal}, isotropic covariance functions on the sphere are isotropic with respect to both chordal and great-arc distance. By adding constraints to the scaling parameters $\gamma_1(\sloc)$ and $\gamma_2(\sloc)$, the correlation function $\rho_{NS}$ in \eqref{equ:general_cov_design} can achieve isotropy on the sphere:

\begin{theorem}
\label{theo:isotropy}
The correlation function $\rho_{NS}$ in \eqref{equ:general_cov_design} is isotropic (i.e., depends only on distance) if $\gamma_1 ({\sloc})=\gamma_2 ({\sloc}) \equiv \gamma$ is constant.
\end{theorem}

A subclass of covariance functions that are especifically useful for spherical domains are axially symmetric covariance functions \citep[e.g.,][]{stein2007spatial}, under which the correlation between a pair of locations on the sphere depends on longitudes only through the longitude difference:

\begin{definition*}
A covariance function $C: \sphere \times \sphere \rightarrow \real$ is called axially symmetric if there exists a function $C_A$ such that
\[
C(\sloc_i,\sloc_j) = C_A(\lon_i -\lon_j, \lat_i, \lat_j),
\]
where $\sloc_i = (\lon_i,\lat_i)$ and $\sloc_j = (\lon_j,\lat_j)$ with longitudes $\lon_i$ and $\lon_j$ and latitudes $\lat_i$ and $\lat_j$ on $\sphere$.
\end{definition*}

Axially symmetric covariane functions can be also obtained based on the Theorem below:

\begin{theorem}
\label{theo:symmetry}
The correlation function $\rho_{NS}$ in \eqref{equ:general_cov_design} is axially symmetric if $\kappa({\sloc}) \equiv 0$ and $\gamma_1(\cdot)$ and $\gamma_2(\cdot)$ are functions of latitude only (i.e., they do not depend on longitude).
\end{theorem}

\begin{figure}
\centering
\hspace*{\fill}
	\begin{subfigure}{.32\textwidth}
	\includegraphics[width=1\linewidth]{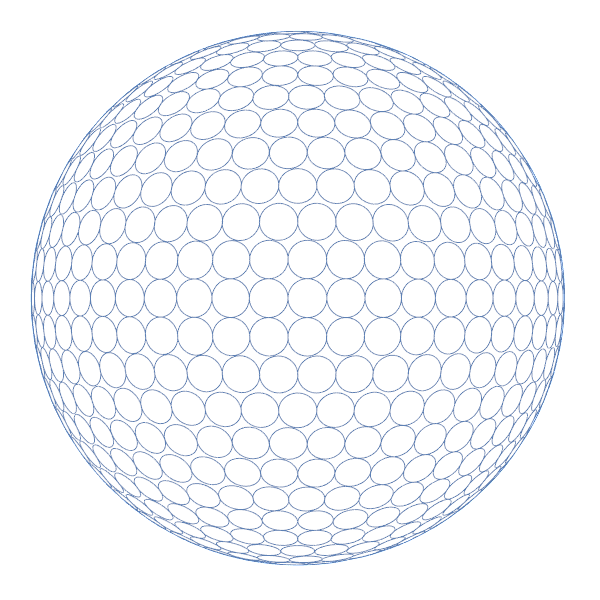}
	\caption{Isotropy ($\gamma_1 = \gamma_2$)}
	\end{subfigure}
	\hfill
	\begin{subfigure}{.32\textwidth}
		\includegraphics[width=1\linewidth]{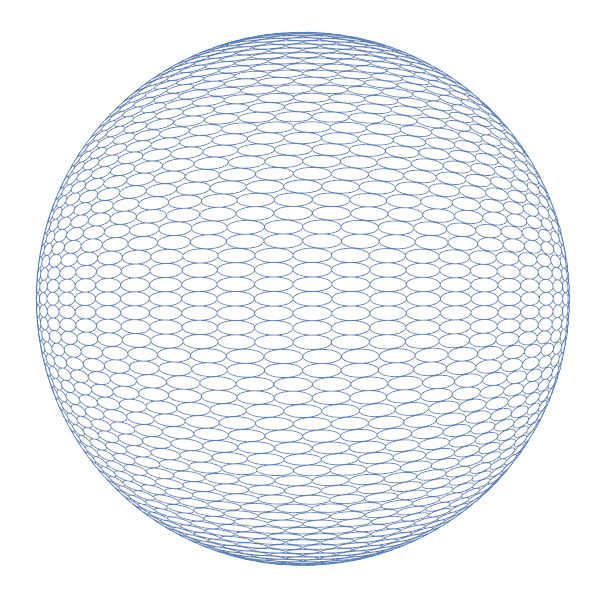}
	\caption{``Anisotropy'' ($\kappa = 0; \gamma_1 > \gamma_2$)}
	\end{subfigure}
\hspace*{\fill}

\hspace*{\fill}
	\begin{subfigure}{.32\textwidth}
	\includegraphics[width=1\linewidth]{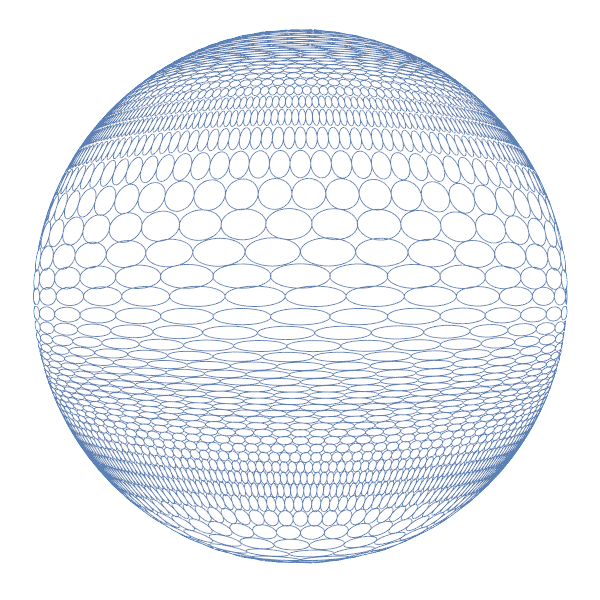}
	\caption{Axial symmetry}
	\end{subfigure}
	\hfill
	\begin{subfigure}{.32\textwidth}
	\includegraphics[width=1\linewidth]{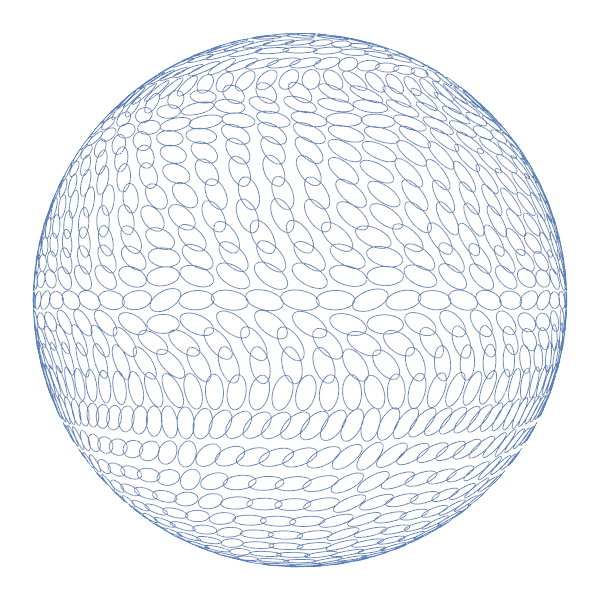}
    \caption{General nonstationarity with spatially varying rotation $\kappa(\cdot)$}
	\end{subfigure}
\hspace*{\fill}
\caption{Illustration of special cases of the nonstationary correlation functions in \eqref{equ:general_cov_design} via correlation contours}
\label{fig:illus_anisotropy}
\end{figure}
Special cases of our general nonstationary covariance function that include isotropic, anisotropic, axially symmetric, and general nonstationary parameterizations are visualized in Figure \ref{fig:illus_anisotropy}.

\subsection{Example: A nonstationary Mat\'{e}rn covariance on the sphere}
\label{subsec:nonstationary_matern}

The Mat\'{e}rn correlation function is highly popular in geospatial analysis. It is valid in $\real^d$ for any $d \in \mathbb{N}$ and given by 
\[
    \matern_\nu(r) = \textstyle\frac{2^{1-\nu}}{\Gamma(\nu)} r^\nu \mathcal{K}_\nu(r), \qquad r \geq 0,
\]
where $ \mathcal{K}_\nu( \cdot) $ is the modified Bessel function of order $\nu >0$. The standard deviation, say  $\sigma(\cdot)$, and the smoothness parameter $\nu(\cdot)$ in the Mat\'ern can also vary over space \citep{Stein2005}. Hence, we can obtain a highly flexible Mat\'ern covariance on the sphere of the form
\begin{equation}
\label{equ:nonstatmatcor}
  \matern_{NS}( \sloc_i, \sloc_j ) = \sigma(\sloc_i)\sigma(\sloc_j)c(\sloc_i,\sloc_j) \matern_{ (\nu(\sloc_i) + \nu(\sloc_j))/2}  (q(\sloc_i,\sloc_j)),
\end{equation}
where $c, q$ are as in \eqref{equ:mahalanobis}--\eqref{equ:cov_const_term}, and $\bfSigma(\bs)$ can be parameterized in terms of spatially varying scales $\bfgamma(\bs)$ and rotation $\rho(\bs)$ as in Section \ref{subsec:construct}.

\citet{Guinness2016} showed that the local smoothness properties of the Mat\'ern covariance are preserved when restricting a process in Euclidean space to the sphere. Specifically, a GP with covariance function $\matern_{NS}$ has $m$ mean square derivatives at $\sloc$ if and only if $\nu(\sloc) > m$. 

\section{Vecchia approximation\label{sec:Vecchia}}

For many modern large datasets, including on the sphere, direct application of GPs is too computationally expensive, as the cost scales cubically in the number of data points. The approximation proposed by \citet{Vecchia1988} has become highly popular in recent years, which has linear computational complexity and straight-forward parallel features while maintaining high accuracy measured by the KL divergence from the true process \citep[e.g.,][]{Guinness2016a,Katzfuss2017a}. Based on a given ordering of the observations, the Vecchia approximation replaces the high-dimensional joint multivariate normal density with a product of univariate conditional normal densities, in which each variable conditions only on a small subset of previous observations in the ordering, amounting to an ordered conditionally independent approximation.

Denote $h(i) = \{1,2,...,i-1\}$ with $h(1) = \emptyset$ and $\by_{h(i)} = (y_1, ..., y_{i-1})'$. Consider a GP $y(\cdot) \sim \GP(0,C)$ on a spatial region $\domain$ with covariance function $C$. The distribution of the observation $\by = (y_1, y_2, ...,y_n)$ is given by 
\[
f(\by) = \prod_{i=1}^{n}   f (y_i |\by_{h(i)}).
\]
The Vecchia approximation replaces $h(i)$ with a subset $g(i) \subset h(i)$, where $g(i)$ is usually chosen to select those indices corresponding to the $m$ observations nearest in distance to the $i$th observation. This leads to the following approximation of the joint density:
\begin{equation}
\hat{f}(\by) = \prod_{i=1}^{n}   f (y_i |\by_{g(i)}). \label{eq:vecchia}    
\end{equation}

The Vecchia approximation ensures computational feasibility for large spatial datasets. The choices for ordering the locations and selecting the conditioning sets $\{g(i)\}_{i = 1}^{n}$ are typically based on distance or estimated correlation \citep{Katzfuss2017a}. Here, we will use maximum-minimum-distance ordering \citep{Guinness2016a} and nearest-neighbor conditioning in the numerical studies, both based on the Euclidean distances. Correlation-based ordering and conditioning that takes into account the potential nonstationary structure is also possible \citep{Katzfuss2020,Kang2021}. Aside from likelihood-based parameter inference based on the Vecchia likelihood in \eqref{eq:vecchia}, the Vecchia approximation can also be applied to unknown locations in order to obtain accurate approximations of posterior predictive distributions \citep[e.g.,][]{Katzfuss2018}. In the case of noisy data, the Vecchia approximation can be applied to the latent (noise-free) GP as before, and then combined with an incomplete-Cholesky decomposition of the posterior precision matrix to preserve the low computational complexity \citep{Schafer2020}.
We implemented Vecchia inference based on our new covariance function by extending the \texttt{R} package \texttt{GPvecchia} \citep{GPvecchia}.

\section{Numerical study\label{sec:simulation}}

We perform simulations to demonstrate the improvement of posterior inference gained from adopting a more flexible covariance structure. Specifically, we simulate GPs that are isotropic, axially symmetric, and generally nonstationary based on different parameterizations of the covariance structure introduced in \eqref{equ:nonstatmatcor}. The scaling parameters $\gamma_1(\sloc)$ and $\gamma_2(\sloc)$ are parameterized as:
\begin{equation}
\gamma_1(\sloc) = \exp(\beta_{10} + \beta_{11} \sin (l) + \beta_{12}L),
\label{eq:gamma1}
\end{equation}
\begin{equation}
\gamma_2(\sloc) = \exp(\beta_{20} + \beta_{21} \sin (l) + \beta_{22}L),
\label{eq:gamma2}
\end{equation}
where $\sloc=(l,L)$ with longitude $l$ and latitude $L$. Based on Theorems~\ref{theo:isotropy} and \ref{theo:symmetry}, isotropic, axially symmetric, and general nonstationary covariance structures are constructed as follows:
\begin{description}
\item[Isotropy:]
According to Theorem \ref{theo:isotropy}, the correlation function $\rho_{NS}$ in \eqref{equ:general_cov_design} is isotropic if $\gamma_1 ({\sloc})=\gamma_2 ({\sloc}) \equiv \gamma$ is constant, and so we set the parameters in \eqref{eq:gamma1}--\eqref{eq:gamma2} as 
\[
\beta_1 =(\beta_{10},\beta_{11},\beta_{12}  )= (-0.5,0,0),
\]
\[
\beta_2 =(\beta_{20},\beta_{21},\beta_{22}  )=(-0.5,0,0).
\]

\item[Axial symmetry:]
According to Theorem \ref{theo:symmetry}, the correlation function $\rho_{NS}$ in \eqref{equ:general_cov_design} is axially symmetric if $\kappa({\sloc}) \equiv 0$ and $\gamma_1(\cdot)$ and $\gamma_2(\cdot)$ are functions of latitude only and we set the parameters as
\[
\beta_1 =(\beta_{10},\beta_{11},\beta_{12}  )= (-0.5,0,1.44),  
\]
\[
\beta_2 =(\beta_{20},\beta_{21},\beta_{22}  )=(-3.2,0,1.44).  
\]

\item[General nonstationarity:]
A more general nonstationary covariance function can be obtained by setting 
\[
\kappa=0.8,
\]
\[
\beta_1 =(\beta_{10},\beta_{11},\beta_{12}  )= (-0.5,-1.2,1.44),  
\]
\[
\beta_2 =(\beta_{20},\beta_{21},\beta_{22}  )=(-3.2,-0.3,1.44).  
\]
\end{description}

Each dataset is generated on a regular latitude/longitude grid of size $50 \times 50 = 2{,}500$ on the sphere and then split into training and testing subsets, under two different scenarios: (a) a random sampling of $20\%$ of all locations as the test dataset; (b) ten randomly selected regions, each with a longitudinal band width of $0.4$ and a latitudinal band width of $0.2$, that sum up to approximately $20\%$ of all locations as the test dataset. The training dataset is then modeled by the three progressively more flexible covariance structures:
\begin{description}
\item[Isotropy:] unknown $\beta_{10},\beta_{20}$, with $\beta_{10} = \beta_{20}$ and fixed $\beta_{11}=\beta_{12}=\beta_{21}=\beta_{22}=\kappa=0$.
\item[Axial symmetry:] unknown $\beta_{10},\beta_{12},\beta_{20},\beta_{22}$, with fixed $\beta_{11}=\beta_{21}=\kappa=0$.
\item[General nonstationarity:] unknown $\beta_{10},\beta_{11},\beta_{12},\beta_{20},\beta_{21},\beta_{22},\kappa$.
\end{description}

Realizations of the isotropic, axially symmetric, and general nonstationary GPs are shown in Figure~\ref{fig:sim_truths}.
\begin{figure}
    \centering
	\begin{subfigure}{.3\textwidth}
	\includegraphics[width=\textwidth,trim=60 70 80 0,clip]{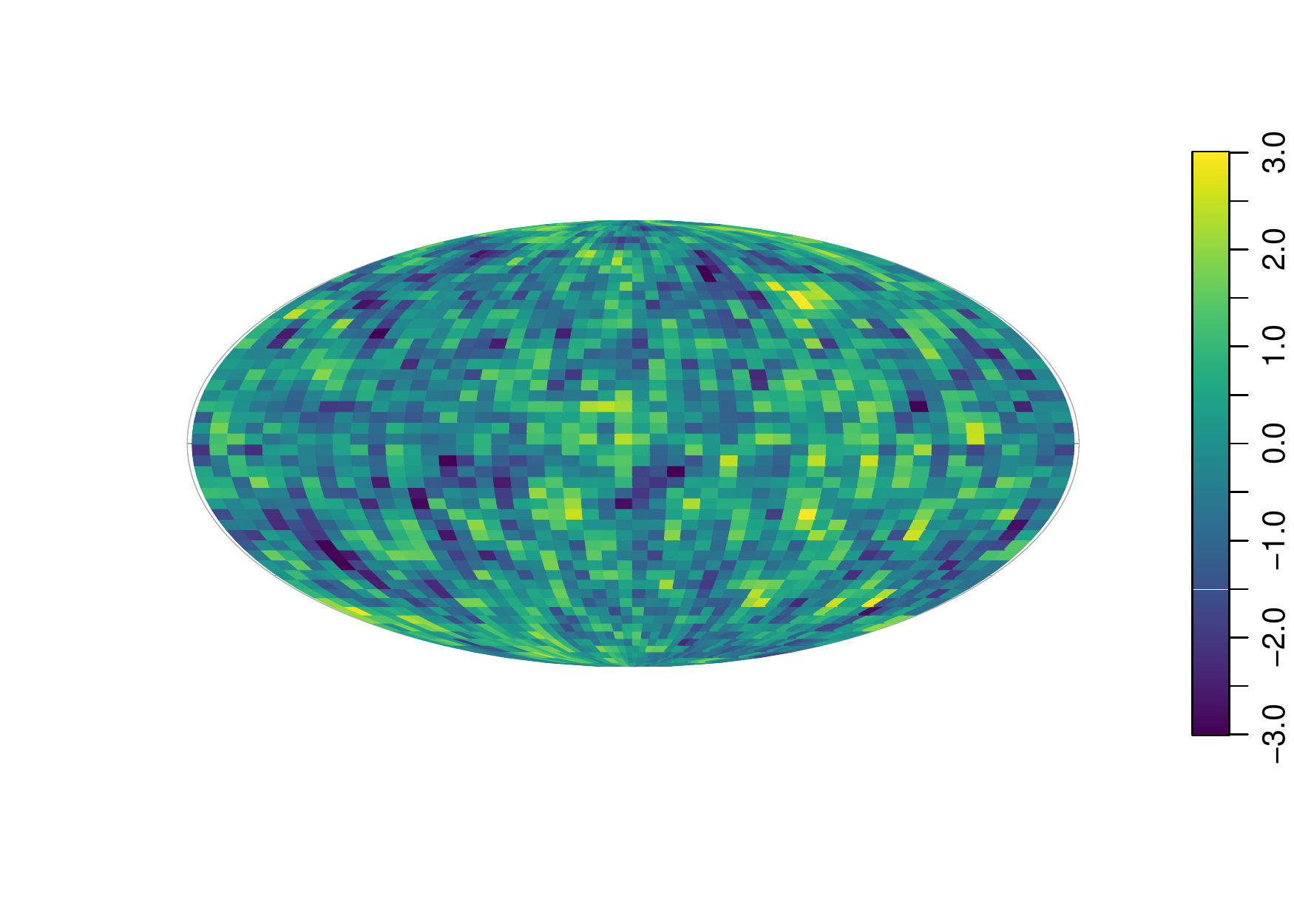}
	\caption{Isotropic}
	\end{subfigure}\hfill
	\begin{subfigure}{.3\textwidth}
	\includegraphics[width=\textwidth,trim=60 70 80 0,clip]{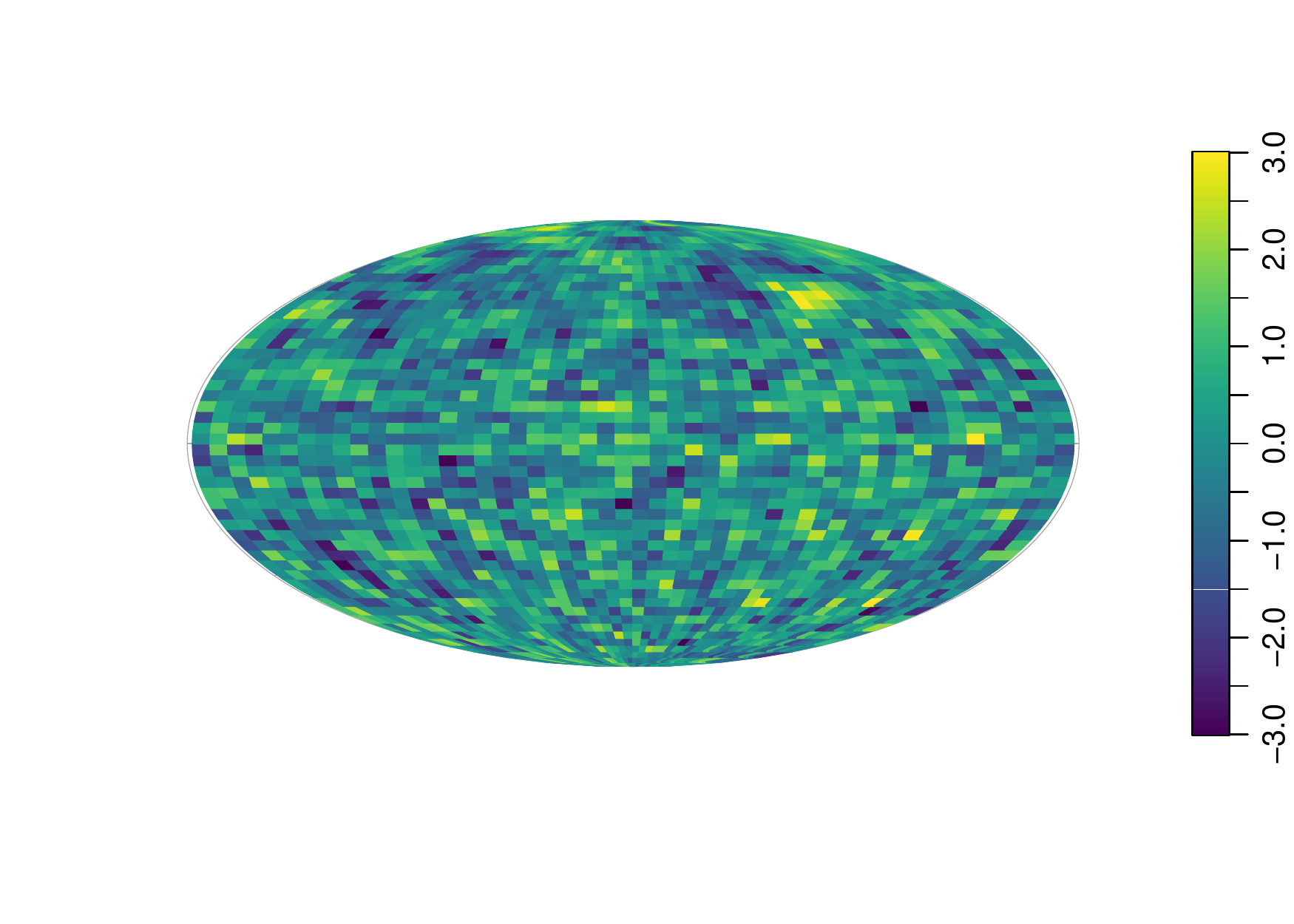}
	\caption{Axially symmetric}
	\end{subfigure}\hfill
	\begin{subfigure}{.3\textwidth}
	\includegraphics[width=\textwidth,trim=60 70 80 0,clip]{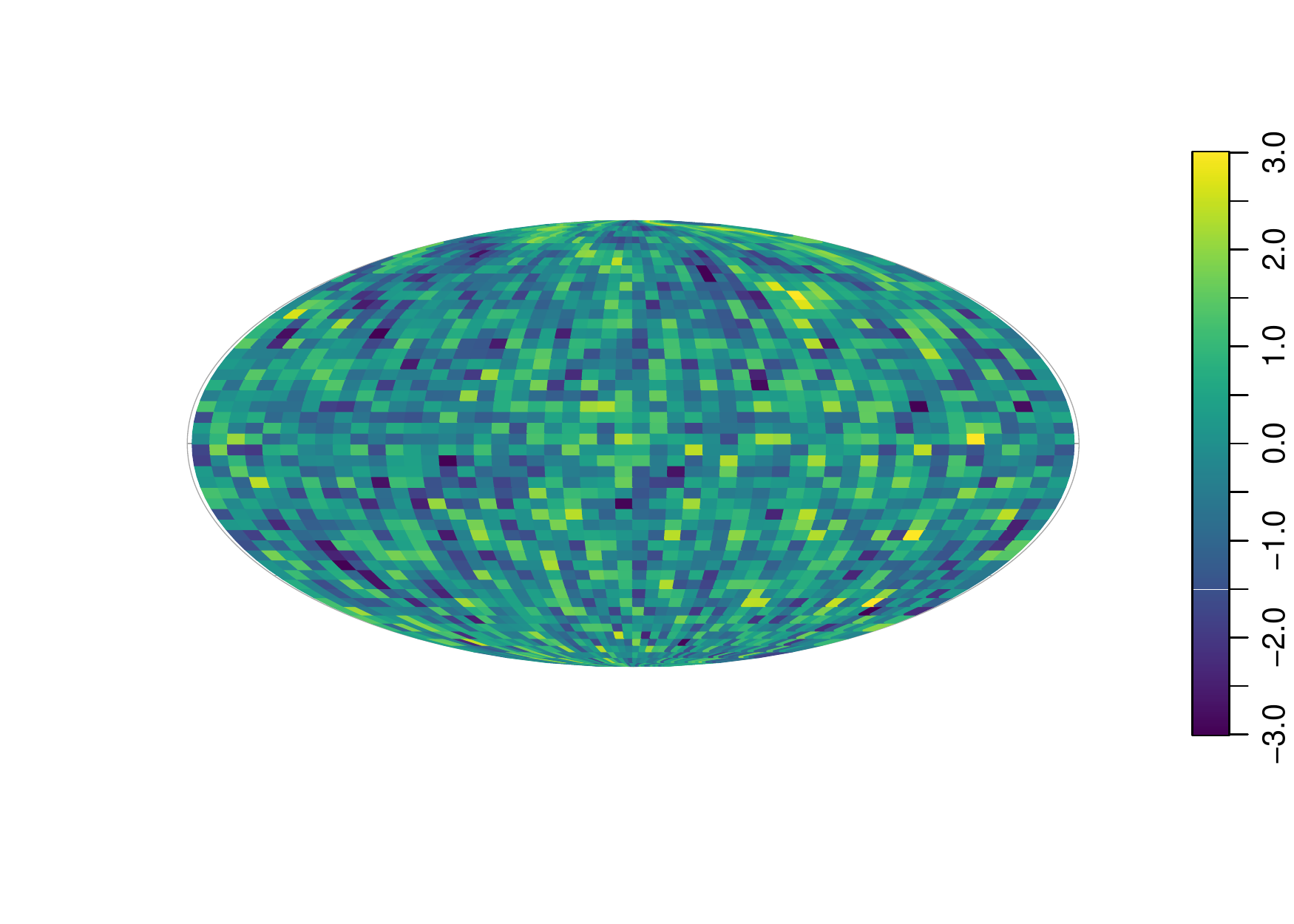}
	\caption{General nonstationary}
	\end{subfigure}\hfill
	\begin{subfigure}{.03\textwidth}
	\flushright
	\includegraphics[height=3cm,trim=450 35 0 35,clip]{plots/nonsta-true.pdf}
	\end{subfigure}
	
	\caption{Realizations of the isotropic, axially symmetric, and general nonstationary GPs over a $50 \times 50$ grid on the sphere}
	\label{fig:sim_truths}
\end{figure}
The prediction scores for the nine combinations of `true' and `assumed' covariance structures are summarized in Table \ref{table:comp} that include the mean absolute error (MAE), the root mean squared error (RMSE), the continuous ranked probability score (CRPS), and the energy score. MAE and RMSE consider only point predictions, CRPS evaluates marginal predictive distributions, and the energy score evaluates the joint predictive distribution for the entire test set. The variance and smoothness parameters of the Mat\'ern covariance function $\matern_\nu(r)$ are $\sigma(\bs) = 1$ and $\nu(\bs) = 0.5$, respectively, both considered as known. We ran an adaptive MCMC algorithm \citep{vihola2012robust} for $5{,}000$ iterations for the Bayesian inference of unknown parameters. Within MCMC, the Vecchia approximation in \eqref{eq:vecchia} of the likelihood is applied, which uses the maximum-minimum-distance ordering and the nearest-neighbor conditioning with a conditioning set size of ten. To obtain the posterior predictive distribution at the testing locations, we also used the Vecchia approximation (i.e., ordinary kriging with ten nearest neighbors), based on MCMC samples of the unknown parameters after a burn-in of size one thousand; the resulting predictive distribution is a mixture of Gaussians. For datasets generated by isotropic GPs, the prediction scores are very close under the assumptions of isotropic, axially symmetric, and nonstationary covariance structures. The difference becomes more pronounced when the datasets are generated from an axially symmetric GP, where the axially symmetric and the nonstationary structures have similar performance, both significantly better than that of the isotropic covariance structure. Furthermore, in modeling the general nonstationary GPs, there is a uniform improvement in all prediction scores when switching from the axially symmetric covariance structure to the general nonstationary covariance structure. Hence, the prediction accuracy is barely affected when using a general flexible model when the true model is simple and the number of optimization parameters is small, but large gains are possible with a more flexible model when the true dependence structure is more complicated.
\begin{table}
\begin{subtable}{\textwidth}
	\centering
    \captionsetup{justification=raggedright,singlelinecheck=false, font=normalsize}
    \caption{True model - Isotropic }
    \setlength{\tabcolsep}{0.18cm}
	\begin{tabular}{l||r|r|r|r||r|r|r|r}
		\hline
		& \multicolumn{4}{c||}{Random} & \multicolumn{4}{c}{Region}\\
		\hline
		& MAE &RMSE &CRPS &Energy &MAE &RMSE &CRPS &Energy\\
		\hline
		Isotropic  & 0.569 & 0.728 & 0.563 & 16.1 & 0.716 & 0.904 & 0.710 & 18.8\\
        Axially symmetric  & 0.567 & 0.727 & 0.556 & 15.9 & 0.716 & 0.904 & 0.705 & 18.6\\
        Nonstationary  & 0.568 & 0.728 & 0.551 & 15.7 & 0.716 & 0.904 & 0.698 & 18.4\\
		\hline
	\end{tabular}
\end{subtable}
\newline
\vspace{0.6cm}
\newline
\begin{subtable}{\textwidth}
	\centering
    \captionsetup{justification=raggedright,singlelinecheck=false, font=normalsize}
    \caption{True model - Axially symmetric }
    \setlength{\tabcolsep}{0.18cm}
	\begin{tabular}{l||r|r|r|r||r|r|r|r}
		\hline
		& \multicolumn{4}{c||}{Random} & \multicolumn{4}{c}{Region}\\
		\hline
		& MAE &RMSE &CRPS &Energy &MAE &RMSE &CRPS &Energy\\
		\hline
		Isotropic  & 0.754 & 0.961 & 0.751 & 21.3 & 0.768 & 0.968 & 0.761 & 20.1\\
        Axially symmetric  & 0.637 & 0.834 & 0.621 & 18.1 & 0.741 & 0.932 & 0.732 & 19.2\\
        Nonstationary  & 0.637 & 0.835 & 0.616 & 18.0 & 0.741 & 0.931 & 0.727 & 19.1\\
		\hline
	\end{tabular}
\end{subtable}
\newline
\vspace{0.6cm}
\newline
\begin{subtable}{\textwidth}
	\centering
    \captionsetup{justification=raggedright,singlelinecheck=false, font=normalsize}
    \caption{True model - Nonstationary }
    \setlength{\tabcolsep}{0.18cm}
	\begin{tabular}{l||r|r|r|r||r|r|r|r}
		\hline
		& \multicolumn{4}{c||}{Random} & \multicolumn{4}{c}{Region}\\
		\hline
		& MAE &RMSE &CRPS &Energy &MAE &RMSE &CRPS &Energy\\
		\hline
		Isotropic  & 0.734 & 0.938 & 0.730 & 20.8 & 0.777 & 0.973 & 0.773 & 20.3\\
        Axially symmetric  & 0.688 & 0.883 & 0.671 & 19.2 & 0.761 & 0.953 & 0.752 & 19.6\\
        Nonstationary  & 0.681 & 0.874 & 0.659 & 18.8 & 0.754 & 0.943 & 0.739 & 19.3\\
		\hline
	\end{tabular}
\end{subtable}
\caption{Prediction scores (lower is better), each averaged over five simulated datasets, for the nine different combinations of the true and assumed covariance structures. Test sets are selected as random locations (Random) or regions (Region) that amount to $20\%$ of the total dataset.}
\label{table:comp}
\end{table}

\section{Application to real data}
\label{sec:real_app}

Using the same three types of covariance structures used in Section~\ref{sec:simulation}, we model a precipitation dataset from the Community Earth System Model (CESM) Large Ensemble Project \citep{Kay2015}. After subsetting, the dataset contains precipitation rates (m/s) on July 1, 401, on a roughly a $2^\circ$ resolution in terms of longitude and latitude, totaling $144 \times 96 = 13{,}824$ locations on a spherical grid. 
We consider the standardized log-precipitation anomalies shown in Figure~\ref{fig:app_prcp}, which does not indicate any distinct mean structure. 
\begin{figure}
    \centering
  	\includegraphics[width = 0.5\textwidth,trim = 0mm 20mm 0mm 0mm, clip]{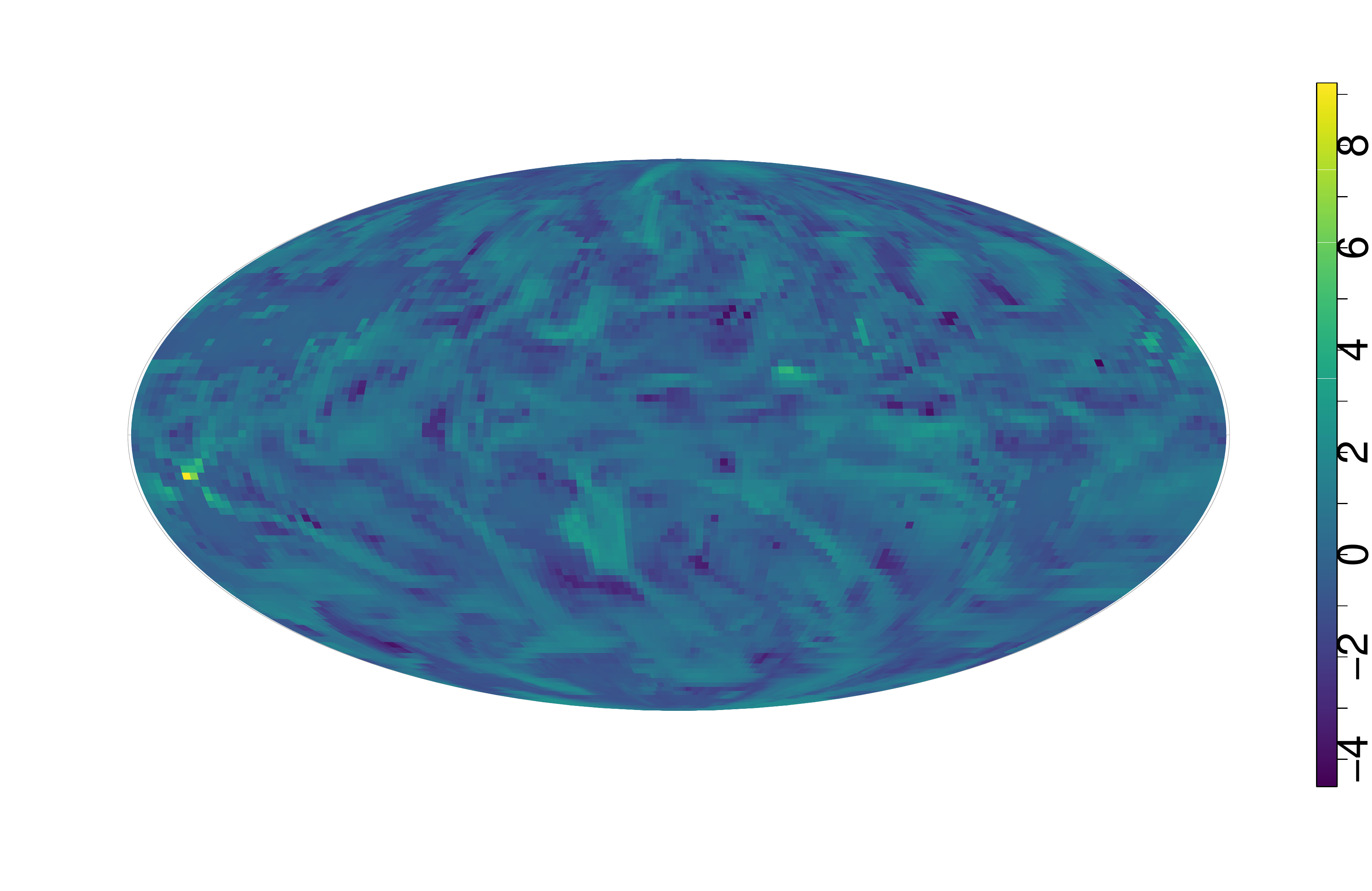}
	\caption{Visualization of the precipitation dataset after pre-processing}
	\label{fig:app_prcp}
\end{figure}
$20\%$ of the dataset is used for testing, selected either as random locations or as random regions, similar to Section~\ref{sec:simulation}. 

Model parameters have the same initializations under the three model assumptions, and the smoothness parameter $\nu$ for the CESM dataset is fixed at $2.5$ due to the increased smoothness compared with the simulated datasets in Figure~\ref{fig:sim_truths}. We also summarize the performance (e.g., scores) of posterior inference in Table~\ref{table:real_data_compare}.
\begin{table}
	\centering
    \setlength{\tabcolsep}{0.18cm}
	\begin{tabular}{l||r|r|r|r||r|r|r|r}
		\hline
		& \multicolumn{4}{c||}{Random} & \multicolumn{4}{c}{Region}\\
		\hline
		& MAE &RMSE &CRPS &Energy &MAE &RMSE &CRPS &Energy\\
		\hline
		Isotropic  & 0.383 & 0.549 & 0.376 & 27.1 & 0.756 & 0.967 & 0.749 & 43.1\\
        Axially symmetric  & 0.268 & 0.426 & 0.264 & 21.6 & 0.654 & 0.834 & 0.649 & 37.7\\
        Nonstationary  & 0.269 & 0.428 & 0.263 & 21.6 & 0.651 & 0.832 & 0.620 & 35.7\\
		\hline
	\end{tabular}
\caption{Prediction scores, using the precipitation data, of isotropic, axially symmetric, and general non-stationary structures. Test sets are selected as random locations (Random) or regions (Region) that amount to $20\%$ of the total dataset.}
\label{table:real_data_compare}
\end{table}
The isotropic covariance structure is significantly out-performed by the axially symmetric and the nonstationary structures while the difference between the latter two is indistinguishable. This indicates that using an isotropic covariance structure can be insufficient for practical modeling and the benefit of having a parsimonious flexible structure is likely to outweigh that of substituting the Euclidean distance with the chordal distance, especially when scalable GP approximations (e.g., the Vecchia approximation) are applied.

\section{Conclusions\label{sec:conclusions}}

We proposed a general approach for constructing nonstationary, locally anisotropic covariance functions on the sphere based on isotropic covariance functions in $\real^3$. Special parameterizations of the nonstationary covariance function that amount to isotropic and axially symmetric covariance structures were also discussed. Axially symmetric covariance functions are widely used in geospatial analysis of global data and their advantages over isotropic covariance structures were demonstrated with both simulated Gaussian random fields and a CESM dataset. The extra flexibility of our nonstationary parameterization also improved posterior inference compared with the axially symmetric structure, although the improvement was less significant. For large datasets on the sphere, straightforward computation of Gaussian probabilities is too computationally expensive, and so we used the Vecchia approximation to achieve faster parameter estimation and posterior inference. 


\footnotesize
\appendix
\section*{Acknowledgments}

MK and JC were partially supported by National Science Foundation (NSF) Grant DMS--1654083. MK was also partially supported by NSF Grant DMS--1521676.

\section{Proofs \label{app:proofs}}

\begin{proof}[Proof of Theorem \ref{theo:isotropy}]
If $\gamma_1 ({\sloc})=\gamma_2 ({\sloc}) \equiv \gamma$ is constant, then
$${\bf D(\gamma)}=\left( \begin{array}{ccc}
1 & 0 & 0 \\
0 & \gamma & 0 \\
0 & 0 & \gamma \end{array} \right),
$$
\begin{align*}
{\bf \tilde{\Sigma}(s)} &= \mathcal{R}_x (\kappa(s)) {\bf D(\gamma(s))} \mathcal{R}_x (\kappa(s))'\\
& =  \left( \begin{array}{ccc}
1 & 0 & 0 \\
0 & \cos \kappa(s) & -\sin \kappa(s) \\
0 & \sin \kappa(s) & \cos \kappa(s) \end{array} \right)
\left( \begin{array}{ccc}
1 & 0 & 0 \\
0 & \gamma & 0 \\
0 & 0 & \gamma \end{array} \right)
\left( \begin{array}{ccc}
1 & 0 & 0 \\
0 & \cos \kappa(s) & \sin \kappa(s) \\
0 & -\sin \kappa(s) & \cos \kappa(s) \end{array} \right) \\
& = \left( \begin{array}{ccc}
1 & 0 & 0 \\
0 & \gamma & 0 \\
0 & 0 & \gamma \end{array} \right).
\end{align*}
To compute ${\bf \Sigma(s)}=\mathcal{R}_z (l) \mathcal{R}_y (L) {\bf \tilde{\Sigma}(s)} \mathcal{R}_y (L)' \mathcal{R}_z (l)'$, we first compute
\begin{align*}
{\bf A}:=\mathcal{R}_z (l) \mathcal{R}_y (L) & = \left( \begin{array}{ccc}
\cos(l) & -\sin(l) & 0 \\
\sin(l) & \cos(l) & 0 \\
0 & 0 & 1 \end{array} \right)
\left( \begin{array}{ccc}
\cos(L) & 0 & \sin(L) \\
0 & 1 & 0 \\
-\sin(L) & 0 & \cos(L) \end{array} \right) \\
& = \left( \begin{array}{ccc}
\cos(l) \cos(L) & -\sin(l) & \cos(l) \sin(L) \\
\sin(l) \cos(L) & \cos(l) & \sin(l) \sin(L) \\
-\sin(L) & 0 & \cos(L) \end{array} \right),
\end{align*}
\begin{align*}
{\bf \Sigma(s)}={\bf A \tilde{\Sigma}(s) A'} & = \left( \begin{array}{ccc}
(1-\gamma)x^2+\gamma & (1-\gamma)xy & (1-\gamma)xz \\
(1-\gamma)xy & (1-\gamma)y^2+\gamma & (1-\gamma)yz \\
(1-\gamma)xz & (1-\gamma)yz & (1-\gamma)z^2+\gamma \end{array} \right) \\
& = (1-\gamma)\left( \begin{array}{c}
x \\
y \\
z \end{array} \right)
\left( \begin{array}{ccc}
x & y & z\end{array} \right) + \gamma {\bf I}_3\\
& = (1-\gamma){\bf \tilde{s} \tilde{s}'} + \gamma {\bf I}_3,
\end{align*}
where $x=\cos(L)\cos(l)$, $y=\cos(L)\sin(l)$, $z=\sin(L)$ are the $(x,y,z)$-coordinates of a 3-dimensional Cartesian coordinate system. Then
\begin{align*}
|{\bf \Sigma(s)}| & = det\{(1-\gamma)\ {\bf \tilde{s} \tilde{s}'} + \gamma {\bf I}_3\}\\
& = \gamma^3 \cdot det\left\{{\bf I}_3 +\frac{1-\gamma}{\gamma}{\bf \tilde{s} \tilde{s}'}\right\}\\
& = \gamma^3 \cdot det\left\{1 +\frac{1-\gamma}{\gamma}{\bf \tilde{s}' \tilde{s}}\right\}\\
& = \gamma^3 \cdot det\left\{1 +\frac{1-\gamma}{\gamma} \cdot 1\right\}\\
& = \gamma^2
\end{align*}
does not depend on $\sloc$. And for $i \neq j$,
$$({\bf \Sigma(s_i)+\Sigma(s_j)})^{-1} = \left((1-\gamma){\bf \tilde{s}_i \tilde{s}_i'} + (1-\gamma){\bf \tilde{s}_j \tilde{s}_j'} + 2\gamma {\bf I}_3 \right)^{-1}.$$
WLOG, we ignore the constant coefficients inside the inverse, and then
\begin{align*}
({\bf \Sigma(s_i)+\Sigma(s_j)})^{-1} & = \left({\bf \tilde{s}_i \tilde{s}_i'} + {\bf \tilde{s}_j \tilde{s}_j'} + {\bf I}_3 \right)^{-1}\\
& = ({\bf \tilde{s}_j \tilde{s}_j'} + {\bf I}_3)^{-1} - ({\bf \tilde{s}_j \tilde{s}_j'} + {\bf I}_3)^{-1} {\bf \tilde{s}_i}[1+{\bf \tilde{s}_i}'({\bf \tilde{s}_j \tilde{s}_j'} + {\bf I}_3)^{-1} {\bf \tilde{s}_i}]^{-1} {\bf \tilde{s}_i}' ({\bf \tilde{s}_j \tilde{s}_j'} + {\bf I}_3)^{-1}.
\end{align*}
Let $\bf B := ({\bf \tilde{s}_j \tilde{s}_j'} + {\bf I}_3)^{-1}$, and so
$$({\bf \Sigma(s_i)+\Sigma(s_j)})^{-1} = {\bf B} - {\bf B} {\bf \tilde{s}_i} (1+{\bf \tilde{s}_i' B \tilde{s}_i})^{-1}{\bf \tilde{s}_i' B}={\bf B} - \frac{{\bf B \tilde{s}_i \tilde{s}_i' B}}{1+{\bf \tilde{s}_i' B \tilde{s}_i}},$$
\begin{align*}
q^2({\sloc_i,s_j}) & \propto ({\bf \tilde{s}_i - \tilde{s}_j})' ({\bf \Sigma(s_i) + \Sigma(s_j)})^{-1} ({\bf \tilde{s}_i - \tilde{s}_j})\\
& \propto ({\bf \tilde{s}_i - \tilde{s}_j})' {\bf B} ({\bf \tilde{s}_i - \tilde{s}_j})-\frac{1}{1+{\bf \tilde{s}_i' B \tilde{s}_i}} ({\bf \tilde{s}_i - \tilde{s}_j})' {\bf B \tilde{s}_i \tilde{s}_i' B} ({\bf \tilde{s}_i - \tilde{s}_j}).
\end{align*}
So computation of $q({\sloc_i,s_j})$ only involves terms ${\bf \tilde{s}_i' B \tilde{s}_i}$, ${\bf \tilde{s}_j' B \tilde{s}_j}$ and ${\bf \tilde{s}_i' B \tilde{s}_j}$. Because
$${\bf B}=({\bf \tilde{s}_j \tilde{s}_j'} + {\bf I}_3)^{-1}={\bf I}_3-{\bf \tilde{s}_j} (1+{\bf \tilde{s}_j' \tilde{s}_j})^{-1} {\bf \tilde{s}_j'}={\bf I}_3-\frac{1}{2}{\bf \tilde{s}_j \tilde{s}_j'},$$
we have
\begin{align*}
    {\bf \tilde{s}_i' B \tilde{s}_i} & = {\bf \tilde{s}_i' \tilde{s}_i} - \frac{1}{2}({\bf \tilde{s}_i' \tilde{s}_j})^2 = 1-\frac{1}{2}({\bf \tilde{s}_i' \tilde{s}_j})^2\\
    {\bf \tilde{s}_j' B \tilde{s}_j} & = {\bf \tilde{s}_j' \tilde{s}_j} - \frac{1}{2}({\bf \tilde{s}_j' \tilde{s}_j})^2 = 1-\frac{1}{2}=\frac{1}{2}\\
    {\bf \tilde{s}_i' B \tilde{s}_j} & = {\bf \tilde{s}_i' \tilde{s}_j} - \frac{1}{2}({\bf \tilde{s}_i' \tilde{s}_j}) ({\bf \tilde{s}_j' \tilde{s}_j}) = {\bf \tilde{s}_i' \tilde{s}_j}-\frac{1}{2}{\bf \tilde{s}_i' \tilde{s}_j}=\frac{1}{2}{\bf \tilde{s}_i' \tilde{s}_j}.
\end{align*}
Further,
$${\bf \tilde{s}_i' \tilde{s}_j} = \left[({\bf \tilde{s}_i - \tilde{s}_j})'({\bf \tilde{s}_i - \tilde{s}_j})-{\bf \tilde{s}_i' \tilde{s}_i} - {\bf \tilde{s}_j' \tilde{s}_j}\right]/2 = \left[({\bf \tilde{s}_i - \tilde{s}_j})'({\bf \tilde{s}_i - \tilde{s}_j})-2\right]/2.$$
So $q({\sloc_i,\sloc_j})$ just depends on the distance $({\bf \tilde{s}_i - \tilde{s}_j})'({\bf \tilde{s}_i - \tilde{s}_j})$. For the normalization term $c({\sloc_i, \sloc_j})$, since we have proved that $|{\bf \Sigma(s_i)}|=|{\bf \Sigma(s_j)}| \equiv \gamma^2$,
\begin{align}
\begin{split}
    c({\sloc_i, \sloc_j}) & = |{\bf \Sigma(s_i)}|^{1/4} |{\bf \Sigma(s_j)}|^{1/4} |({\bf \Sigma(s_i)}+{\bf \Sigma(s_j)})/2|^{-1/2}\\
    & \propto |({\bf \Sigma(s_i)}+{\bf \Sigma(s_j)})^{-1}|^{1/2}\\
    & \propto \left(|({\bf \Sigma(s_i)}+{\bf \Sigma(s_j)})^{-1}| \cdot |({\bf \tilde{s}_i - \tilde{s}_j})'({\bf \tilde{s}_i - \tilde{s}_j})|/[({\bf \tilde{s}_i - \tilde{s}_j})'({\bf \tilde{s}_i - \tilde{s}_j})] \right)^{1/2}\\
    & \propto \left(|({\bf \Sigma(s_i)}+{\bf \Sigma(s_j)})^{-1}| \cdot det\{({\bf \tilde{s}_i - \tilde{s}_j})({\bf \tilde{s}_i - \tilde{s}_j})'\}/[({\bf \tilde{s}_i - \tilde{s}_j})'({\bf \tilde{s}_i - \tilde{s}_j})] \right)^{1/2}\\
    & \propto \left(det\left\{({\bf \Sigma(s_i)}+{\bf \Sigma(s_j)})^{-1} ({\bf \tilde{s}_i - \tilde{s}_j})({\bf \tilde{s}_i - \tilde{s}_j})'\right\}/[({\bf \tilde{s}_i - \tilde{s}_j})'({\bf \tilde{s}_i - \tilde{s}_j})] \right)^{1/2}\\
    & \propto \left(det\left\{({\bf \tilde{s}_i - \tilde{s}_j})' ({\bf \Sigma(s_i)}+{\bf \Sigma(s_j)})^{-1} ({\bf \tilde{s}_i - \tilde{s}_j})\right\}/[({\bf \tilde{s}_i - \tilde{s}_j})'({\bf \tilde{s}_i - \tilde{s}_j})] \right)^{1/2}\\
    & \propto \left\{ \frac{q^2({\sloc_i,\sloc_j})}{({\bf \tilde{s}_i - \tilde{s}_j})'({\bf \tilde{s}_i - \tilde{s}_j})} \right\}^{1/2}.
\end{split}\label{eq:cproof}
\end{align}
We have proved that $q({\sloc_i,\sloc_j})$ just depends on $({\bf \tilde{s}_i - \tilde{s}_j})'({\bf \tilde{s}_i - \tilde{s}_j})$, so $c({\sloc_i, \sloc_j})$ also only depends on the distance $({\bf \tilde{s}_i - \tilde{s}_j})'({\bf \tilde{s}_i - \tilde{s}_j})$.\\
Overall, we can show
$$\rho_{NS}({\sloc_i,\sloc_j})=c({\sloc_i,\sloc_j})\rho(q({\sloc_i,\sloc_j}))$$ only depends on the distance $({\bf \tilde{s}_i - \tilde{s}_j})'({\bf \tilde{s}_i - \tilde{s}_j})$, where $\rho(q)$ is a valid isotropic correlation function. So $\rho_{NS}({\sloc_i,\sloc_j})$ is isotropic.
\end{proof}

\begin{proof}[Proof of Theorem \ref{theo:symmetry}]
If $\kappa({\sloc}) \equiv 0$ and $\gamma_1(\cdot)$, $\gamma_2(\cdot)$ depend on $L$ only, then $\mathcal{R}_x (\kappa(s)) \equiv \mathcal{R}_x (0) = {\bf I}_3$. Then
$${\bf \tilde{\Sigma}(s)} = {\bf D(\gamma(s))}=\left( \begin{array}{ccc}
1 & 0 & 0 \\
0 & \gamma_1(L) & 0 \\
0 & 0 & \gamma_2(L) \end{array} \right)
=\left(\begin{array}{ccc}
1 & 0 & 0 \\
0 & \gamma_1(L) & 0 \\
0 & 0 & \gamma_1(L) \end{array} \right)+
\left( \begin{array}{ccc}
0 & 0 & 0 \\
0 & 0 & 0 \\
0 & 0 & \gamma_2(L)-\gamma_1(L) \end{array} \right).$$
Due to the results in Theorem \ref{theo:isotropy}, we have
$${\bf \Sigma}({\sloc})=\mathcal{R}_z (l) \mathcal{R}_y (L) {\bf \tilde{\Sigma}(s)} \mathcal{R}_y (L)' \mathcal{R}_z (l)'=(1-\gamma_1(L)){\bf \tilde{s} \tilde{s}'} + \gamma_1(L) {\bf I}_3 + (\gamma_2(L)-\gamma_1(L)) {\bf \tilde{s}^*} ({\bf \tilde{s}^*})',$$
where 
$$
{\bf \tilde{s}^*}=\left( \begin{array}{c}
\cos(l) \sin(L) \\
\sin(l) \sin(L) \\
\cos(L) \end{array} \right)
\ \ , \ \
({\bf {\tilde{s}^*}})' ({\bf \tilde{s}^*})=1.
$$
Thus
\begin{align*}
    |{\bf \Sigma(s)}| & = det\{(1-\gamma_1(L))\ {\bf \tilde{s} \tilde{s}'} + \gamma_1(L) {\bf I}_3+ (\gamma_2(L)-\gamma_1(L)) {\bf \tilde{s}^*} ({\bf \tilde{s}^*})'\}\\
    & = \gamma_1(L)^3 \cdot det \left\{ \frac{1-\gamma_1(L)}{\gamma_1(L)} {\bf \tilde{s} \tilde{s}'} + \frac{\gamma_2(L)-\gamma_1(L)}{\gamma_1(L)}{\bf \tilde{s}^*} ({\bf \tilde{s}^*})'+{\bf I}_3 \right\}\\
    & = \gamma_1(L)^3 \cdot det \left\{
    \left( \begin{array}{cc}
       \frac{1-\gamma_1(L)}{\gamma_1(L)} {\bf \tilde{s}} & \frac{\gamma_2(L)-\gamma_1(L)}{\gamma_1(L)}{\bf \tilde{s}^*} \end{array} \right)
    \left(\begin{array}{c}
       {\bf \tilde{s}'} \\
       ({\bf \tilde{s}^*})' \end{array} \right)
        +{\bf I}_3 \right\}\\
    & = \gamma_1(L)^3 \cdot det \left\{
    \left(\begin{array}{c}
       {\bf \tilde{s}'} \\
       ({\bf \tilde{s}^*})' \end{array} \right)
    \left( \begin{array}{cc}
       \frac{1-\gamma_1(L)}{\gamma_1(L)} {\bf \tilde{s}} & \frac{\gamma_2(L)-\gamma_1(L)}{\gamma_1(L)}{\bf \tilde{s}^*} \end{array} \right) +{\bf I}_2 \right\}\\
    & = \gamma_1(L)^3 \cdot \left| \begin{array}{cc}
       \frac{1-\gamma_1(L)}{\gamma_1(L)}+1 & 2\sin(L)\cos(L)\\ 2\sin(L)\cos(L) & \frac{\gamma_2(L)-\gamma_1(L)}{\gamma_1(L)}+1
       \end{array} \right|\\
    & = \gamma_1(L)\gamma_2(L) - 4\gamma_1(L)^3 \sin^2(L) \cos^2(L)
\end{align*}
only depend on $L$. WLOG, ignore $\gamma_1(L)$, $\gamma_2(L)$ again (they only depend on $L$),
\begin{align*}
({\bf \Sigma(s_i)+\Sigma(s_j)})^{-1} & = \left[{\bf \tilde{s}_i \tilde{s}_i'} + {\bf \tilde{s}_i^*} ({\bf \tilde{s}_i^*})' + {\bf \tilde{s}_j \tilde{s}_j'} + {\bf \tilde{s}_j^*} ({\bf \tilde{s}_j^*})' + {\bf I}_3 \right]^{-1}\\
& = \left[ \left( \begin{array}{cc}
       {\bf \tilde{s}_i} & {\bf \tilde{s}_i^*} \end{array} \right)
    \left(\begin{array}{c}
       {\bf \tilde{s}_i'} \\
       ({\bf \tilde{s}_i^*})' \end{array} \right) + 
      \left( \begin{array}{cc}
       {\bf \tilde{s}_j} & {\bf \tilde{s}_j^*} \end{array} \right)
    \left(\begin{array}{c}
       {\bf \tilde{s}_j'} \\
       ({\bf \tilde{s}_j^*})' \end{array} \right) + {\bf I}_3 \right]^{-1}\\
& = {\bf V}^{-1} - {\bf V}^{-1} \left(\begin{array}{cc}
       {\bf \tilde{s}_i} & {\bf \tilde{s}_i^*} \end{array} \right)
       \left[ {\bf I}_2 +  \left(\begin{array}{c}
       {\bf \tilde{s}_i'} \\
       ({\bf \tilde{s}_i^*})' \end{array} \right) {\bf V}^{-1} \left(\begin{array}{cc}
       {\bf \tilde{s}_i} & {\bf \tilde{s}_i^*} \end{array} \right) \right]^{-1}
       \left(\begin{array}{c}
       {\bf \tilde{s}_i'} \\
       ({\bf \tilde{s}_i^*})' \end{array} \right) {\bf V}^{-1}\\
& = {\bf V}^{-1} - {\bf V}^{-1} \left(\begin{array}{cc}
       {\bf \tilde{s}_i} & {\bf \tilde{s}_i^*} \end{array} \right)
       \left[ {\bf I}_2 +  \left(\begin{array}{cc}
       {\bf \tilde{s}_i'} {\bf V}^{-1} {\bf \tilde{s}_i} & {\bf \tilde{s}_i'} {\bf V}^{-1} {\bf \tilde{s}_i^*} \\
       ({\bf \tilde{s}_i^*})' {\bf V}^{-1} {\bf \tilde{s}_i} & ({\bf \tilde{s}_i^*})' {\bf V}^{-1} {\bf \tilde{s}_i^*} \end{array} \right) \right]^{-1}
       \left(\begin{array}{c}
       {\bf \tilde{s}_i'} \\
       ({\bf \tilde{s}_i^*})' \end{array} \right) {\bf V}^{-1},
\end{align*}
where 
$$
{\bf V} =  \left( \begin{array}{cc}
       {\bf \tilde{s}_j} & {\bf \tilde{s}_j^*} \end{array} \right)
    \left(\begin{array}{c}
       {\bf \tilde{s}_j'} \\
       ({\bf \tilde{s}_j^*})' \end{array} \right) + {\bf I}_3.
$$
Then
\begin{align*}
q^2({\sloc_i,\sloc_j}) \propto & ({\bf \tilde{s}_i}-{\bf \tilde{s}_j})' ({\bf \Sigma(s_i)+\Sigma(s_j)})^{-1} ({\bf \tilde{s}_i}-{\bf \tilde{s}_j})\\
\propto & ({\bf \tilde{s}_i}-{\bf \tilde{s}_j})' {\bf V}^{-1} ({\bf \tilde{s}_i}-{\bf \tilde{s}_j}) - ({\bf \tilde{s}_i}-{\bf \tilde{s}_j})' {\bf V}^{-1} \left(\begin{array}{cc}
       {\bf \tilde{s}_i} & {\bf \tilde{s}_i^*} \end{array} \right)
       \left[ {\bf I}_2 +  \left(\begin{array}{cc}
       {\bf \tilde{s}_i'} {\bf V}^{-1} {\bf \tilde{s}_i} & {\bf \tilde{s}_i'} {\bf V}^{-1} {\bf \tilde{s}_i^*} \\
       ({\bf \tilde{s}_i^*})' {\bf V}^{-1} {\bf \tilde{s}_i} & ({\bf \tilde{s}_i^*})' {\bf V}^{-1} {\bf \tilde{s}_i^*} \end{array} \right) \right]^{-1}\cdot\\
       & \left(\begin{array}{c}
       {\bf \tilde{s}_i'} \\
        ({\bf \tilde{s}_i^*})' \end{array} \right) {\bf V}^{-1}({\bf \tilde{s}_i}-{\bf \tilde{s}_j}).
\end{align*}
Because
\begin{align*}
{\bf V}^{-1} & = {\bf I}_3 - \left(\begin{array}{cc}
       {\bf \tilde{s}_j} & {\bf \tilde{s}_j^*} \end{array} \right)
       \left[ {\bf I}_2 + \left(\begin{array}{cc}
       1 & {\bf \tilde{s}_j}' {\bf \tilde{s}_j^*}\\
       ({\bf \tilde{s}_j})' {\bf \tilde{s}_j} & 1 \end{array} \right) \right]^{-1}
       \left(\begin{array}{c}
       {\bf \tilde{s}_j'} \\
       ({\bf \tilde{s}_j^*})' \end{array} \right)\\
& = {\bf I}_3 - \frac{1}{4-({\bf \tilde{s}_j}' {\bf \tilde{s}_j^*})^2} \left(\begin{array}{cc}
       {\bf \tilde{s}_j} & {\bf \tilde{s}_j^*} \end{array} \right)
       \left[ \left(\begin{array}{cc}
       2 & -{\bf \tilde{s}_j}' {\bf \tilde{s}_j^*}\\
       -({\bf \tilde{s}_j})' {\bf \tilde{s}_j} & 2 \end{array} \right) \right]^{-1}
       \left(\begin{array}{c}
       {\bf \tilde{s}_j'} \\
       ({\bf \tilde{s}_j^*})' \end{array} \right),
\end{align*}
we can figure out that the computation of $q^2({\sloc_i, \sloc_j})$ only involves the following types of terms
$$
\left\{
\begin{aligned}
{\bf \tilde{s}_i}' {\bf \tilde{s}_i}  = & 1\\
{\bf \tilde{s}_i}' {\bf \tilde{s}_i^*}  =  & {\bf \tilde{s}_i}' \cdot \left[ (\tan(s_{i2})){\bf \tilde{s}_i} +  \left(0,0,\cos(s_{i2})-\frac{\sin^2(s_{i2})}{\cos(s_{i2})}\right)' \right] = \tan(s_{i2}) + \sin(s_{i2})\left[\cos(s_{i2})-\frac{\sin^2(s_{i2})}{\cos(s_{i2})}\right]\\
({\bf \tilde{s}_i^*})' {\bf \tilde{s}_i}^*  = & 1\\
({\bf \tilde{s}_i^*})' {\bf \tilde{s}_j^*}  =  &\left[ (\tan(s_{i2})){\bf \tilde{s}_i} +  \left(0,0,\cos(s_{i2})-\frac{\sin^2(s_{i2})}{\cos(s_{i2})} \right)' \right]' \cdot \left[ (\tan(s_{j2})){\bf \tilde{s}_j} +  \left(0,0,\cos(s_{j2})-\frac{\sin^2(s_{j2})}{\cos(s_{j2})} \right)' \right]\\
        = & \tan(s_{i2}) \tan(s_{j2}) ({\bf \tilde{s}_i}' {\bf \tilde{s}_j}) + \tan(s_{i2}) \sin(s_{i2})\left[\cos(s_{j2})-\frac{\sin^2(s_{j2})}{\cos(s_{j2})}\right] \\
       & + \tan(s_{j2}) \sin(s_{j2})         \left[\cos(s_{i2})-\frac{\sin^2(s_{i2})}{\cos(s_{i2})}\right]
        + \left[\cos(s_{i2})-\frac{\sin^2(s_{i2})}{\cos(s_{i2})}\right]\left[\cos(s_{j2})-\frac{\sin^2(s_{j2})}{\cos(s_{j2})}\right]\\
{\bf \tilde{s}_i}' {\bf \tilde{s}_j^*}  =  &\tan(s_{j2})({\bf \tilde{s}_i}' {\bf \tilde{s}_j}) + \sin(s_{i2})\left[\cos(s_{j2})-\frac{\sin^2(s_{j2})}{\cos(s_{j2})}\right]\\
{\bf \tilde{s}_j}' {\bf \tilde{s}_i^*}  =  &\tan(s_{i2})({\bf \tilde{s}_i}' {\bf \tilde{s}_j}) + \sin(s_{j2})\left[\cos(s_{i2})-\frac{\sin^2(s_{i2})}{\cos(s_{i2})}\right]\\
{\bf \tilde{s}_i}' {\bf \tilde{s}_j}  =  &\left[({\bf \tilde{s}_i - \tilde{s}_j})'({\bf \tilde{s}_i - \tilde{s}_j})-2\right]/2.
\end{aligned}
\right.
$$
We can change the index $i$ to $j$ for the first 3 terms and they are still valid. Thus these values only depend on $s_{i2}$, $s_{j2}$ and ${\bf \tilde{s}_i}' {\bf \tilde{s}_j}$, and ${\bf \tilde{s}_i}' {\bf \tilde{s}_j}$ can be expressed in terms of the distance $({\bf \tilde{s}_i - \tilde{s}_j})'({\bf \tilde{s}_i - \tilde{s}_j})$. The computation of $q^2({\sloc_i,\sloc_j})$ only depends on the distance $({\bf \tilde{s}_i - \tilde{s}_j})'({\bf \tilde{s}_i - \tilde{s}_j})$ and the longitudes $s_{i2}$, $s_{j2}$. Similar to \eqref{eq:cproof} in the proof of Theorem \ref{theo:isotropy}, we can also show that $c({\sloc_i,\sloc_j})$ is a function of $({\bf \tilde{s}_i - \tilde{s}_j})'({\bf \tilde{s}_i - \tilde{s}_j})$, $s_{i2}$ and $s_{j2}$. Then $\rho_{NS}({\sloc_i,\sloc_j})=c({\sloc_i,\sloc_j})\rho(q({\sloc_i,\sloc_j})):=\rho_A ({\bf \tilde{s}_i - \tilde{s}_j}, s_{i2}, s_{j2})$, so it is axially symmetric.
\end{proof}

\footnotesize
\bibliographystyle{apalike}
\bibliography{mendeley, additionalrefs}

\end{document}